\documentclass[12pt]{iopart}

\usepackage{iopams}  
\usepackage{slashbox}
\usepackage{multirow}
\usepackage{pstricks}

\def\slantfrac#1#2{\hbox{$\,^#1\!/_#2$}}

\newcommand{\demi}{\frac{1}{2}}
\newcommand{\prm}{PRB}
\newcommand{\prms}{\prm{}s}
\newcommand{\eprm}{$\varepsilon$-\prm}
\newcommand{\eprms}{\eprm{}s}

\newcommand{\dprm}{$\delta$-\prm}
\newcommand{\dprms}{\dprm{s}}

\newcommand{\modify}[1]{\normalfont{#1}}
\newcommand{\pns}[1]{P_{\modify{ns},{#1}}}
\newcommand{\pnse}{P_{\modify{ns}}}
\newcommand{\ploc}{P_{\modify{local}}}
\newcommand{\pld}[1]{P_{\modify{ld,{#1}}}}

\newcommand*{\ep}{\varepsilon}

\newcommand{\st}{\mathrm{s.t.}}

\newtheorem{definition}{Definition}
\newtheorem{lemma}{Lemma}
\newtheorem{proposition}{Proposition}
\newenvironment{proof}{\noindent{\em Proof.}}{}

\newtheorem{theorem}{Theorem}

\newcommand{\binom}[2]{{{#1} \choose {#2}}}

\newcommand{\qqed}{\hfill$\Box$}

\begin{document}

\title{The non-locality of $n$ noisy Popescu-Rohrlich boxes}

\author{\small{Matthias Fitzi$^1$, Esther H\"anggi$^1$, 
Valerio Scarani$^{2}$ and Stefan Wolf$^1$}}

\address{$^1$ Computer Science Department, ETH Zurich, Switzerland}
\address{$^2$ Centre for Quantum Technologies and Department of Physics, National University of Singapore, Singapore}
\ead{esther.haenggi@inf.ethz.ch}

\begin{abstract}
We quantify the amount of {\em non-locality\/} contained in $n$ noisy versions of so-called
{\em Popescu-Rohrlich  boxes (PRBs)}, i.e., bipartite systems violating the CHSH Bell inequality 
maximally. Following the approach by Elitzur, Popescu, and Rohrlich, we measure the amount of non-locality
of a system by representing it as a convex combination of a {\em local\/} behaviour, with
maximal possible weight, and a {\em non-signalling\/} system. We show that the local part
of $n$ systems, each of which approximates a PRB with probability $1-\ep$, is of order $\Theta(\ep^{\lceil n/2\rceil})$
in the isotropic, and equal to $(3\ep)^n$ in the maximally biased case.
\end{abstract}

\section{Introduction}
The behaviour of a bipartite input/output system $P_{XY|UV}$ is
{\em non-local\/} if it cannot be obtained from pre-shared information.
For example, the measurement-choice/outcome behaviour of certain
{\em entangled\/} quantum states is non-local in this sense, as first proved by Bell \cite{bellInequality}. This notion of non-locality has applications in device-independent quantum cryptography \cite{kent,acin,masanesc,ec}, randomness extraction \cite{colbeck,pironio}, and state estimation~\cite{bardyn}. Non-local correlations 
can also be seen as a resource to fulfill distributed tasks~\cite{vandam}.

There is not a unique way of quantifying non-locality: For instance, the efficiency of a resource in performing any of the above-mentioned tasks can be taken as a measure of non-locality with operational meaning. Colbeck and Renner, for example, quantify the non-locality of a system by the possibility of giving biased outputs as seen from an adversaries point of view~\cite{colren} and show that quantum theory has \emph{no local part} in this sense. A different measure is based on an idea by Elitzur, Popescu, and Rohrlich \cite{EPR2}. It consists of partitioning the probability distribution into a local part of maximal weight and the remaining non-local part. So far, this measure of non-locality has been studied on a few examples of systems that can be obtained by measurements on quantum states \cite{bkp,scarani,branciard}.
We will see that this measure yields lower bounds  
in contexts where a system is used to realize an information-theoretic task,
and where the success and strength of this simulation depends on the amount of 
(non-)locality in the system. Since (noisy) PR boxes are a simple and natural type of 
non-local systems easily obtained in nature, we quantify the non-locality of a number 
of them by this measure. 

 The {\em Popescu-Rohrlich box (PRB)\/} is a hypothetical device that, on binary inputs $X$ and $Y$, produces random binary outputs $U$ and $V$ such that $X\oplus Y=U\cdot V$ \cite{pr,rastall,khalfin}. Imperfect (or noisy) \prms{} are probability distributions that fulfill 
this condition only with a certain probability, i.e.,  $Pr[X\oplus Y=U\cdot V]=1-\ep$ for random inputs. For an imperfect \prm{} achievable with quantum resources the error $\ep$ must be at least $\sim 15 \%$~\cite{tsirelson}; if only classical resources are available, it must be at least $25\%$.

In this paper, we address the following question: How much does non-locality, measured by decomposing the system into a \emph{local} and a \emph{non-local} part, increase if one has $n$ copies of an imperfect \prm{}? We prove that the {\em local\/} part decreases exponentially in the number of copies. More precisely, the local part
of $n$ systems, each of which approximates a PRB with probability $1-\ep$, is in $\Theta(\ep^{\lceil n/2\rceil})$
in the isotropic and $(3\ep)^n$ in the maximally biased case.
Our result is closely related to  the distillability of \prms{} and of no-signalling resources in general \cite{short,zurich,brunner}, 
to device-independent QKD~\cite{masanesc,ec}.

\section{Definitions}

A {\em bipartite input-output  system\/}  takes an input and yields
an output from a well-defined alphabet on each side (i.e., to each {\em party\/}) and
is completely characterized by a {\em conditional probability distribution\/}
$P_{XY|UV}(x,y,u,v)$, where $U$ and $V$ are the inputs, and $X$ and $Y$ are the
outputs, respectively.
We restrict our considerations to {\em bipartite\/} systems; however,  generalizations
to more parties are possible. 

All distributions we consider in this paper are {\em non-signalling}, i.e., do not allow for
message transmission. The existence of systems not having this property would be in sharp
contrast to relativity theory as soon as the two inputs could be given, and the outputs obtained,
in a space-like separated fashion.

\begin{definition}
{\rm
 A bipartite conditional probability distribution $P_{XY|UV}(x,y,u,v)$ is called
 \emph{non-signalling} if the two parties cannot use it to transmit information,
 i.e., 
 \begin{eqnarray}
\nonumber  \sum_x P_{XY|UV}(x,y,u,v) &= \sum_x P_{XY|UV}(x,y,u',v)\ \mbox{\ for all\ }  y,v\ ,\\
\nonumber  \sum_y P_{XY|UV}(x,y,u,v) &= \sum_y P_{XY|UV}(x,y,u,v')\ \mbox{\ for all\ }  x,u\ .
 \end{eqnarray}
}
\end{definition}

The space of all non-signalling probability distributions of a certain
input/output alphabet is a convex polytope. 

\begin{definition}
{\rm
A non-signalling probability distribution is \emph{local deterministic} if it can
be written as
 \begin{equation}
 \nonumber
  P_{XY|UV}=\delta_{x,f(u)}\cdot \delta_{y,g(v)}\ ,
 \end{equation}
where $f:U\rightarrow X$ and $g:V\rightarrow Y$ are deterministic functions
mapping from the set of inputs to the set of outputs and $\delta$ is the Kronecker symbol
defined by $\delta_{xy}:=1$ if $x=y$ and $\delta_{xy}:=0$ otherwise. A non-signalling probability distribution is \emph{local} if it is a convex
combination of local deterministic probability distributions. 
}
\end{definition}

Intuitively, local determinism  means  that each output is uniquely determined by simply the input on this side.
All {\em local\/} probability distributions can be simulated by two distant parties 
with shared randomness. The latter  indicates
 which local deterministic probability distribution to use. Altogether,  the output is
then a deterministic function of the randomness plus the  input on the same side. 

\begin{definition}{\bf (Elitzur, Popescu, Rohrlich~\cite{EPR2})}
{\rm
Given a bipartite non-signalling probability distribution $P_{XY|UV}$,
the maximum $p$, $0\leq p \leq 1$, such that the probability distribution can be
written as the convex combination of a local system, with weight $p$, and a non-signalling system,
of weight $1-p$, 
 is called its \emph{local part}:
 \begin{equation}
 \nonumber
  P_{XY|UV}=p\cdot\ploc+(1-p)\cdot\pnse\ .
 \end{equation}
}
\end{definition}

A probability distribution is local if and only if its local part is $1$.
In the special case of probability distributions taking binary inputs and
giving binary outputs, there is a simple inequality which can be used to determine
if a probability distribution is local. 
\begin{proposition}[Bell~\cite{bellInequality}]\label{bell}
 A bipartite probability distribution $P_{XY|UV}$ taking binary input and giving
 binary output is non-local if 
 \begin{equation}
 \nonumber
  Pr[X\oplus Y=U\cdot V]>0.75
 \end{equation}
 for uniform inputs. 
\end{proposition}

Note that up to relabelling of the inputs and outputs, the above condition is,
actually, \emph{equivalent} to non-locality. After~\cite{CHSH},
 the condition $X\oplus Y=U\cdot V$ is called the \emph{CHSH condition}. 
For the general case of  larger input and output alphabets,  Lemma~\ref{lemma:can_occur_with_p}
will be useful. 

\begin{lemma}\label{lemma:can_occur_with_p}
 Consider two non-signalling probability distributions $P_{XY|UV}$ and $\pns{1}$.
 Then the former  can be written as a convex combination of the latter  with
 weight $p$ and a  non-signalling probability distribution $\pns{2}$
with weight $1-p$, i.e., 
 \begin{equation}
 \nonumber
  P_{XY|UV}=p\cdot\pns{1}+(1-p)\cdot\pns{2}\ ,
 \end{equation}
 if and only if 
 \begin{equation}\label{eq:lem1cond}
 \nonumber
  p\cdot\pns{1}(x,y,u,v) \leq  P_{XY|UV}(x,y,u,v)\ \ \mbox{holds for all\ } x,y,u,v\ .
 \end{equation}
\end{lemma}
\noindent
In particular, this holds if $\pns{1}$ is  local deterministic. 

\begin{proof} 
{\rm
Assume first that $p\cdot\pns{1}(x,y,u,v)\leq P_{XY|UV}(x,y,u,v)\ \mbox{\ for all\ }  x,y,u,v$.
Since both $P_{XY|UV}$ and $\pns{1}$, are normalized and non-signalling, 
\begin{equation}\nonumber
\pns{2}:=\frac{P_{XY|UV}-p\cdot \pns{1}}{1-p}
\end{equation}
is also
 normalized and non-signalling since both properties are linear. 
Finally, Condition~(\ref{eq:lem1cond}) implies that $\pns{2}$ is non-negative. 

 To see the reverse direction, assume $p\cdot\pns{1}(x,y,u,v)> P_{XY|UV}(x,y,u,v)$ for some
 $x,y,u,v$. Then $\pns{2}(x,y,u,v)<0$, and  $\pns{2}$ is not a probability distribution.\qqed
}
\end{proof}

\section{Isotropic \eprms{}}
\subsection{One Isotropic \eprm{}}
We now study the case of one single \eprm{} ($\varepsilon\in [0,0.25]$), i.e., a \prm{} that
fulfills the CHSH condition with probability $1-\varepsilon$ for each input pair, and for which the output bits
on both sides are unbiased, given the input pair. 

\begin{definition}
{\rm
 An \emph{isotropic \eprm{}} is a bipartite conditional probability distribution given by the
 following probability table.  
\begin{eqnarray}
\label{1box_unbiased_table}
\begin{array}{c c||c|c||c|c||}
$\backslashbox{V}{U}$& & \multicolumn{2}{c||}{0} & \multicolumn{2}{c||}{1} \\
 & $\backslashbox{Y}{X}$ & 0 & 1 & 0 & 1 \\ \hline\hline
\multirow{2}{*}{0} & 0 & \frac{1}{2}-\frac{\varepsilon}{2} & \frac{\varepsilon}{2} & \frac{1}{2}-\frac{\varepsilon}{2} & \frac{\varepsilon}{2} \\ \cline{2-6}
& 1 & \frac{\varepsilon}{2} & \frac{1}{2}-\frac{\varepsilon}{2} & \frac{\varepsilon}{2} & \frac{1}{2}-\frac{\varepsilon}{2} \\ \hline\hline
\multirow{2}{*}{1} & 0 & \frac{1}{2}-\frac{\varepsilon}{2} & \frac{\varepsilon}{2} & \frac{\varepsilon}{2} & \frac{1}{2}-\frac{\varepsilon}{2} \\ \cline{2-6}
& 1 & \frac{\varepsilon}{2} & \frac{1}{2}-\frac{\varepsilon}{2} & \frac{1}{2}-\frac{\varepsilon}{2} & \frac{\varepsilon}{2} \\ \hline\hline
\end{array}
\end{eqnarray}
We denote this probability distribution by $P^{1,\varepsilon}_{XY|UV}$ (for $1$ \eprm{}). 
}
\end{definition}
An isotropic \eprm{} can be seen as the convex combination of a perfect \prm{} and a
completely random bit:
\begin{equation}\nonumber
P^{1,\varepsilon}_{XY|UV} = 2\varepsilon\cdot P^{1,1/2}_{XY|UV}+(1-2\varepsilon)\cdot P^{1,0}_{XY|UV}\ .
\end{equation}
Note that the distribution of the random bit is completely local, i.e., its local
part is equal to $1$ while the perfect \prm{}'s local part is $0$.
However, the conclusion that the local part of $P^{1,\varepsilon}_{XY|UV}$ must be $2\varepsilon$
is wrong because $P^{1,\varepsilon}_{XY|UV}$ can be expressed as another convex combination
with higher local weight as follows. 
\begin{small}
\begin{eqnarray}\label{eq:best_dec_1_symm_box}
 %\begin{array}{rcr}
\nonumber
\fl P^{1,\varepsilon}_{XY|UV}
&=&
+\frac{\varepsilon}{2}\cdot
\begin{array}{c c||c|c||c|c||}
& & \multicolumn{2}{c||}{0} & \multicolumn{2}{c||}{1} \\
 &
& 0 & 1 & 0 & 1 \\ \hline\hline
\multirow{2}{*}{0} & 0 & 1 & 0 & 1 & 0 \\ \cline{2-6}
& 1 & 0 & 0 & 0 & 0 \\ \hline\hline
\multirow{2}{*}{1} & 0 & 1 & 0 & 1 & 0 \\ \cline{2-6}
& 1 & 0 & 0 & 0 & 0 \\ \hline\hline
\end{array}
+
\frac{\varepsilon}{2}\cdot
\begin{array}{c c||c|c||c|c||}
& & \multicolumn{2}{c||}{0} & \multicolumn{2}{c||}{1} \\
 &
& 0 & 1 & 0 & 1 \\ \hline\hline
\multirow{2}{*}{0} & 0 & 1 & 0 & 1 & 0 \\ \cline{2-6}
& 1 & 0 & 0 & 0 & 0 \\ \hline\hline
\multirow{2}{*}{1} & 0 & 0 & 0 & 0 & 0 \\ \cline{2-6}
& 1 & 1 & 0 & 1 & 0 \\ \hline\hline
\end{array}
\\ \nonumber
&&
+\frac{\varepsilon}{2}\cdot
\begin{array}{c c||c|c||c|c||}
& & \multicolumn{2}{c||}{0} & \multicolumn{2}{c||}{1} \\
 & 
& 0 & 1 & 0 & 1 \\ \hline\hline
\multirow{2}{*}{0} & 0 & 1 & 0 & 0 & 1 \\ \cline{2-6}
& 1 & 0 & 0 & 0 & 0 \\ \hline\hline
\multirow{2}{*}{1} & 0 & 1 & 0 & 0 & 1 \\ \cline{2-6}
& 1 & 0 & 0 & 0 & 0 \\ \hline\hline
\end{array}
+
%\\ \nonumber
%&&+
\frac{\varepsilon}{2}\cdot
\begin{array}{c c||c|c||c|c||}
& & \multicolumn{2}{c||}{0} & \multicolumn{2}{c||}{1} \\
 & 
& 0 & 1 & 0 & 1 \\ \hline\hline
\multirow{2}{*}{0} & 0 & 0 & 1 & 1 & 0 \\ \cline{2-6}
& 1 & 0 & 0 & 0 & 0 \\ \hline\hline
\multirow{2}{*}{1} & 0 & 0 & 0 & 0 & 0 \\ \cline{2-6}
& 1 & 0 & 1 & 1 & 0 \\ \hline\hline
\end{array}
\\ \nonumber
&&+
%+
\frac{\varepsilon}{2}\cdot
\begin{array}{c c||c|c||c|c||}
& & \multicolumn{2}{c||}{0} & \multicolumn{2}{c||}{1} \\
 & 
& 0 & 1 & 0 & 1 \\ \hline\hline
\multirow{2}{*}{0} & 0 & 0 & 0 & 0 & 0 \\ \cline{2-6}
& 1 & 1 & 0 & 0 & 1 \\ \hline\hline
\multirow{2}{*}{1} & 0 & 1 & 0 & 0 & 1 \\ \cline{2-6}
& 1 & 0 & 0 & 0 & 0 \\ \hline\hline
\end{array}
+
\frac{\varepsilon}{2}\cdot
\begin{array}{c c||c|c||c|c||}
& & \multicolumn{2}{c||}{0} & \multicolumn{2}{c||}{1} \\
 & 
& 0 & 1 & 0 & 1 \\ \hline\hline
\multirow{2}{*}{0} & 0 & 0 & 0 & 0 & 0 \\ \cline{2-6}
& 1 & 0 & 1 & 0 & 1 \\ \hline\hline
\multirow{2}{*}{1} & 0 & 0 & 1 & 0 & 1 \\ \cline{2-6}
& 1 & 0 & 0 & 0 & 0 \\ \hline\hline
\end{array}
\\ \nonumber
&&+
\frac{\varepsilon}{2}\cdot
\begin{array}{c c||c|c||c|c||}
& & \multicolumn{2}{c||}{0} & \multicolumn{2}{c||}{1} \\
 & 
& 0 & 1 & 0 & 1 \\ \hline\hline
\multirow{2}{*}{0} & 0 & 0 & 0 & 0 & 0 \\ \cline{2-6}
& 1 & 0 & 1 & 1 & 0 \\ \hline\hline
\multirow{2}{*}{1} & 0 & 0 & 0 & 0 & 0 \\ \cline{2-6}
& 1 & 0 & 1 & 1 & 0 \\ \hline\hline
\end{array}
+
\frac{\varepsilon}{2}\cdot
\begin{array}{c c||c|c||c|c||}
& & \multicolumn{2}{c||}{0} & \multicolumn{2}{c||}{1} \\
 & 
& 0 & 1 & 0 & 1 \\ \hline\hline
\multirow{2}{*}{0} & 0 & 0 & 0 & 0 & 0 \\ \cline{2-6}
& 1 & 0 & 1 & 0 & 1 \\ \hline\hline
\multirow{2}{*}{1} & 0 & 0 & 0 & 0 & 0 \\ \cline{2-6}
& 1 & 0 & 1 & 0 & 1 \\ \hline\hline
\end{array}
\\ \nonumber
&&+
(1-4\varepsilon)\cdot
\begin{array}{c c||c|c||c|c||}
& & \multicolumn{2}{c||}{0} & \multicolumn{2}{c||}{1} \\
 & 
& 0 & 1 & 0 & 1 \\ \hline\hline
\multirow{2}{*}{0} & 0 & \slantfrac{1}{2} & 0 & \slantfrac{1}{2} & 0 \\ \cline{2-6}
& 1 & 0 & \slantfrac{1}{2} & 0 & \slantfrac{1}{2} \\ \hline\hline
\multirow{2}{*}{1} & 0 & \slantfrac{1}{2} & 0 & 0 & \slantfrac{1}{2} \\ \cline{2-6}
& 1 & 0 & \slantfrac{1}{2} & \slantfrac{1}{2} & 0 \\ \hline\hline
\end{array} 
% \end{array}
\end{eqnarray}
\end{small}
This shows that the local part is at least $4\varepsilon$. In fact, the
local part cannot be found by greedily subtracting local deterministic strategies,
but must be optimized using a {\em linear-programming technique}. By
Lemma~\ref{lemma:can_occur_with_p}, we can write any non-signalling probability
distribution as
\begin{equation}
\nonumber
 P_{XY|UV}=\sum_i{p_i\cdot\pld{i}}+\left(1-\sum_i p_i\right)\cdot\pnse{}\ ,
\end{equation}
where $\pld{i}$ are the different local-deterministic strategies fixed by the
input and output size. Together with the definition of the local part, this
implies the following.

\begin{lemma}\label{lemma:lp}
The local part is the optimal value of the following linear program:
\begin{eqnarray}
\nonumber \max: && \sum_i p_i \\
\nonumber \st & &\sum_i p_i\cdot \pld{i}(x,y,u,v)\leq P_{XY|UV}(x,y,u,v)\\
\nonumber && p_i \geq 0 \ .
\end{eqnarray}
\end{lemma}

This way it can also be shown that the above decomposition of $P^{1,\varepsilon}_{XY|UV}$
is indeed optimal, and that the local part of $P^{1,\varepsilon}_{XY|UV}$ is $4\varepsilon$.

\subsection{Two Isotropic \eprms{}}
Now, consider the system composed of 
two independent isotropic \eprms{}. We can write these two boxes
as one single system taking two input bits $u=(u_1,u_2)$,  $v=(v_1,v_2)$ and giving two output bits $x=(x_1,x_2)$,  $y=(y_1,y_2)$ on each side:
\begin{eqnarray}
\nonumber P^{2,\varepsilon}_{XY|UV}(x,y,u,v)&=&P^{2,\varepsilon}_{XY|UV}({(x_1x_2),(y_1y_2),(u_1u_2),(v_1v_2)})\\
&=&P^{1,\varepsilon}_{XY|UV}(x_1,y_1,u_1,v_1)\cdot P^{1,\varepsilon}_{XY|UV}(x_2,y_2,u_2,v_2)\ .
\end{eqnarray}

Obviously, it is always possible to write each of the two boxes separately as a
combination of one local and one non-local box. This would give a local part of weight
of $(4\varepsilon)^2$. However,  the  local part might be larger and,
actually, it is.
More precisely, it is equal to the local part of
one single isotropic \eprm{}.

\begin{lemma}\label{lemma:two_symm_boxes}
\[
%\begin{equation}
\nonumber P^{2,\varepsilon}_{XY|UV}=(4\varepsilon)\cdot P^{2,local}_{XY|UV}+(1-4\varepsilon)\cdot P^{2,0}_{XY|UV}. 
%\end{equation}
\]
\end{lemma}
\begin{proof}
{\rm
The local part of weight $4\varepsilon$ consists of $128$ local deterministic strategies, which are the following. 
The first $64$ can be obtained from the following, using depolarization~\cite{masanes}:
\begin{eqnarray}
\nonumber u_1u_2 \rightarrow x_1x_2: && 00 \mapsto 00,\ 
01 \mapsto 00,\ 
10 \mapsto 00,\ 
11 \mapsto 01\\ 
\nonumber 
v_1v_2 \rightarrow y_1y_2: && 00 \mapsto 00,\ 
01 \mapsto 00,\ 
10 \mapsto 10,\ 
11 \mapsto 00
\end{eqnarray}
The weight of each of these strategies is $\ep/16-\ep^2/8$.

The second set of $64$ strategies are the ones equal to the following, again using 
depolarization:
\begin{eqnarray}
\nonumber u_1u_2 \rightarrow x_1x_2: && 00 \mapsto 00,\ 
01 \mapsto 00,\ 
10 \mapsto 00,\ 
11 \mapsto 01\\ 
\nonumber 
v_1v_2 \rightarrow y_1y_2: && 00 \mapsto 00,\ 
01 \mapsto 00,\ 
10 \mapsto 00,\ 
11 \mapsto 10
\end{eqnarray}
The weight of each of those strategies is $\ep^2/8$. Together, this yields a local part of weight $4\ep$.\qqed
}
\end{proof}

Lemma~\ref{lemma:two_symm_boxes}
has certain interesting consequences. It has been observed~\cite{kent}
that non-locality can be used for {\em device-independent QKD}. The principal mechanism is as follows: 
A perfect PR box' outputs, if the system is also non-signalling, must be perfectly unbiased. Similarly,
the outputs of an \eprm\  (as long as it is non-local) cannot be {\em completely\/} determined. 
When this fact is interpreted as being from the point of view of an all-powerful adversary 
Eve that is limited by the non-signalling condition only, it leads to the fact that this adversary 
cannot perfectly know the outputs. Clearly, this can potentially be useful cryptographically,
namely for key agreement the security of which is independent of what quantum systems the devices
manipulate on --- actually, independent even from quantum physics as a whole! Now, protocols
for the {\em  amplification of such ``non-signalling secrecy''\/} have been found (under certain additional conditions)~\cite{masanesc,ec}. 
Note, however, that
Lemma~\ref{lemma:two_symm_boxes} implies that {\em two\/} isotropic \eprm\  are, in the above set-up, 
not stronger than a single one: No more secrecy can be extracted. The reason is that the size of the 
local part corresponds to a lower bound on a possible adversary's knowledge on any function of the boxes' outputs.

Similarly, there is a direct connection to non-locality distillation. 
Do {\em two} isotropic \eprms{} allow for the construction
of an $\ep'$-PRB with smaller error $\varepsilon'<\varepsilon$ by applying a function to its inputs and outputs? Again, the answer is negative
by Lemma~\ref{lemma:two_symm_boxes},
 since a
local probability distribution always remains local even when a function is applied to it.
An even stronger result for two systems was shown directly in~\cite{short}, while
in~\cite{dejan} it was shown that, for any number of systems which are realizable by quantum mechanics, the possibility of distilling isotropic \eprms\ 
is at best very limited, and {\em completely impossible\/} for many values of~$\ep$.

\subsection{$n\geq 3$ Isotropic \eprms{}}
Before we consider the general case, let us take a closer look at the case $n=3$. 
We can look at the problem from a {\em game\/} point of view: The \prm{} can be seen as
a tool which always wins the so-called CHSH game. In this game, Alice and Bob are both
given a random bit, and each of them need to reply with a bit.
They win the game if and only if the XOR of their output bits is equal to the AND of
their inputs. Three \prms{} can now be seen as the same game, of which Alice and Bob are
playing three rounds  in parallel. This allows them to apply a better
strategy than when playing each of the three games independently.

\begin{lemma}\label{lemma:3strat_loses2}
For every local deterministic strategy for three \prms{} there always exist inputs
$u$ and $v$ such that Alice and Bob lose two out of the three rounds of the CHSH game.
\end{lemma}
\begin{proof}
{\rm
By contradiction.
Let, without loss of generality,  $x(000)=000$. In order to lose at most one out
of the three rounds of the game for the case $u=000$, $y(v)$ must have Hamming weight
at most one, i.e.,  $y(v)\in \{000,001,010,100\},\mbox{\ for all\ }  v$.
Now, consider $x(111)=x_1x_2x_3\in \{000,001,010,011,100,101,110,111\}$.
Now, consider $y(\bar{x}_1\bar{x}_2\bar{x}_3)=y_1y_2y_3$ when $u=111$: In order to win
all three rounds for this case, it must hold that $y_i=x_i$ if and only if $x_i=1$, i.e.,
$y(\bar{x}_1\bar{x}_2\bar{x}_3)=111$, since
$y_i=u_i\cdot v_i\oplus x_i=v_i\oplus x_i=\bar{x}_i\oplus x_i=1$.
% $u_i+v_i=\bar{x}_i=x_i+y_i$.
Thus, in order to win at least two rounds,
$y(\bar{x}_1\bar{x}_2\bar{x}_3)$ must have Hamming weight at least two. This
contradicts the fact that $y(v)$ must have Hamming weight at most one.\qqed
}
\end{proof}

\begin{lemma}\label{lemma:lds_maxweight}
 Every local deterministic strategy for three \eprms{} can have weight  at most
 $({\varepsilon}/{2})^2(1/2-\varepsilon/2)$.
\end{lemma}
\begin{proof} 
{\rm
 Lemma~\ref{lemma:3strat_loses2} states that for every local deterministic strategy $P_{\mbox{ld}}$,
 there exist $u,v,x,y$ such that $P_{\mbox{ld}}(x,y,u,v)=1$, but
 $P^{3,\varepsilon}_{XY|UV}(x,y,u,v)=({\varepsilon}/{2})^2(1/2-\varepsilon/2)$, since such values
 of $x,y,u,v$ lose two rounds of the CHSH game.
 Together with Lemma~\ref{lemma:can_occur_with_p}, this implies 
 $p\leq ({\varepsilon}/{2})^2(1/2-\varepsilon/2)$.\qqed
}
\end{proof}

\begin{lemma}
The local part of three isotropic \eprms{} is of order $\Theta(\varepsilon^2)$.
\end{lemma}
\begin{proof} 
{\rm
The local part is in $\Omega(\varepsilon^2)$ because the combination of a common
strategy for the first two boxes and a separate strategy for the third box
leads to a local strategy of weight $(4\varepsilon)^2$.
On the other hand, the local part
cannot be larger: The set of strategies which lose at most two rounds of the
CHSH game does not depend on $\varepsilon$  --- hence,
the number of involved local deterministic strategies is 
constant, say  $d$. Thus,  Lemma~\ref{lemma:lds_maxweight}
implies that 
 the local part
is at most $d\cdot ({\varepsilon}/{2})^2(1/2-\varepsilon/2)+O(\varepsilon^3)=O(\varepsilon^2)$.\qqed
}
\end{proof}

\begin{lemma}\label{lemma:n_strat_loses_half}
For every local deterministic strategy for $n$ \prms{}, there always exist inputs $u$ and
$v$ such that Alice and Bob lose at least half of the $n$ rounds of the CHSH game.
\end{lemma}
\begin{proof} 
{\rm
This proof is by contradiction, and a direct generalization of the proof of Lemma~\ref{lemma:3strat_loses2}.
Let, without loss of generality,  $x(0\dots0)=0\dots0$. In order to lose at most
$k(<n)$ out of the $n$ rounds of the game for the case $u=0\dots0$, $y(v)$ must have Hamming
weight at most $k$ (independently of $v$).
Now, let $x(1\dots1)=x_i$, and consider $y(\bar{x}_i)$: In order to win all $n$ rounds for
the case $u=1\dots1$, $y(\bar{x}_i)$ must be equal to  $x(1\dots1)$ exactly at the
positions where $x_i=1$, i.e., $y(\bar{x}_i)=1\dots1$. Thus, in order to lose at
most $k$ rounds, $y(\bar{x}_i)$ must have Hamming weight at least $n-k$. Since $k<n/2$,
this contradicts the fact that $y(v)$ must have Hamming weight at most $k$.\qqed
}
\end{proof} 

\begin{theorem}\label{theorem:local_part_order_n_half}
 The local part of $n$ isotropic \eprms{} is of order
 $\Theta(\varepsilon^{\lceil \frac{n}{2} \rceil})$.
\end{theorem}
\begin{proof} 
{\rm
 First, the local part is in $\Omega(\varepsilon^{\lceil \frac{n}{2} \rceil})$ since
 a local part of weight
 $(4\varepsilon)^{\lceil \frac{n}{2} \rceil}$ can be achieved by combining the
 \eprms{} in pair, and using Lemma~\ref{lemma:two_symm_boxes}. 
 On the other hand, Lemma~\ref{lemma:n_strat_loses_half} states that it cannot be larger
 than $d\cdot \varepsilon^{\lceil \frac{n}{2} \rceil}+O(\varepsilon^{\lceil \frac{n}{2} \rceil+1})$,
 where $d$ is a constant (in $\varepsilon$).\qqed
}
\end{proof}

\subsection{Explicit Bounds}

We now know the asymptotic order of the local part of $n$ \eprms. We are  interested
in determining the quantity precisely and explicitly.

\begin{lemma}\label{lemma:piece_wise_polynomial}
 The local part of $n$ isotropic \eprms{} as a function $f(\varepsilon)$
 is  continuous.
There exist  a finite partition of
 the function domain into intervals $I=\{I_1,\dots,I_m\}$ and 
  polynomials of degree at most $n$, $p_1(\varepsilon),\dots,p_m(\varepsilon)$, such that
 $f(\varepsilon)=p_i(\varepsilon)$ if $\varepsilon\in I_i$.
\end{lemma}
\begin{proof}
{\rm
The local part is determined by the solution of a linear program of the form
\begin{eqnarray}
\nonumber \max : && c^Tx \\
\nonumber \st & &A\cdot x \leq b\\
\nonumber &&x \geq 0\ ,
\end{eqnarray}
where the vector $c$ is the all-one vector, the matrix $A$ only contains $0$s  and $1$s,
and $b$ is a vector of polynomials in $\varepsilon$ (and all other elements
do not depend on this parameter). By definition of the dual program, the solution
 of the above program is equal to the solution of a linear program:
\begin{eqnarray}
\nonumber \min : && b^Ty \\
\nonumber \st& &A^T\cdot y \leq c\\
\nonumber &&y \geq 0\ .
\end{eqnarray}
The domain of this linear program is constant (because none of the inequalities depend
on $\varepsilon$) and a convex set (in fact, a polyhedron). We know that the optimum is
necessarily attained in an extremal point --- a vertex of the polyhedron. Every
vertex corresponds to one specific $y$ --- let us call it $y^k$ for the $k$-th vertex.
The solution of the linear program can then be written as
$\min_k(b^Ty^k)=\min_k(\sum_i b_i\cdot y_i^k)$. As the $b_i$'s are all polynomials in
$\varepsilon$ of degree at most $n$ and the $y_i^k$ are constants, $b^Ty^k$ is a linear
combination of polynomials of order $n$ and, therefore, itself a polynomial of at most
order $n$. Hence, the local part is given by the minimum of a finite number of fixed polynomials of 
degree at most $n$.\qqed
}
\end{proof} 

\begin{lemma}\label{lemma:upperbound}
 The local part of $n$ \eprms{} is at most
 $2^{2n}\cdot \sum_{i=\lceil \frac{n}{2}\rceil}^{n}\binom{n}{i}(1-\varepsilon)^{n-i}\varepsilon^i$.
\end{lemma}
\begin{proof} 
{\rm
 Lemma \ref{lemma:n_strat_loses_half} states that if we sum over all the probability entries 
in the isotropic \eprm{} we have also counted the weight of every
 local strategy at least once. The entry with probability
 $(\frac{1}{2}-\frac{\varepsilon}{2})^{n-i}(\frac{\varepsilon}{2})^i$ occurs exactly
 $2^n\cdot \binom{n}{i}$ times per input and there are $2^{2n}$ inputs. Note that
 this is approximately equal to $(64\varepsilon)^{n/2}$ for large $n$ and small
 $\varepsilon$.\qqed
}
\end{proof} 

For  small  enough $\varepsilon$, the entry with the lowest probability is always the limiting
one. We can, therefore, approximate the leading coefficient better for  small $\varepsilon$
(i.e., the polynomial $f_1$ from below):

\begin{lemma}
 For small enough $\varepsilon$ the local part is at least
 $2^{n/2}\binom{n}{n/2}(1-\varepsilon)^{n/2}\varepsilon^{n/2}$ for even $n$ and
 $2^{(n+3)/2}\binom{n}{(n+1)/2}(1-\varepsilon)^{(n-1)/2}\varepsilon^{(n+1)/2}$ for odd $n$.
\end{lemma}
\begin{proof} 
{\rm
 For the case of two \prms{}, we have seen that there exists a local deterministic strategy which for $8$ different
inputs  wins both rounds of the CHSH game, for another $8$ inputs it wins one round and it never
loses both rounds of the CHSH game.
 Taking the product of this strategy gives us a strategy which never loses more than
 $\lceil \frac{n}{2}\rceil $ rounds (independently of the input) and loses exactly
 $\lceil \frac{n}{2}\rceil $ rounds for exactly $8^{\lfloor \frac{n}{2}\rfloor}$ inputs.
 Through depolarization~\cite{masanes} of this strategy we can obtain a local
 strategy, such that each of the entries with probability
 $(\frac{1}{2}-\frac{\varepsilon}{2})^{n-\lceil \frac{n}{2}\rceil }(\frac{\varepsilon}{2})^{\lceil \frac{n}{2}\rceil }$
 is covered the same number of times and if we sum over all the entries with probability
 $(\frac{1}{2}-\frac{\varepsilon}{2})^{n-\lceil \frac{n}{2}\rceil }(\frac{\varepsilon}{2})^{\lceil \frac{n}{2}\rceil }$,
 then every strategy is counted exactly $8^{\lfloor \frac{n}{2}\rfloor}$ times (note that
 for small epsilon, the limiting probability is always the one with the highest degree in
 $\varepsilon$, no matter how the rest of the strategy looks like). Therefore, for low
 $\varepsilon$, these strategies can reach a local part of weight
 \begin{equation}\nonumber
  8^{-\lfloor n/2 \rfloor}\cdot 2^{2n}\cdot 2^n\cdot \binom{n}{\lceil n/2  \rceil} \left(\frac{1}{2}-\frac{\varepsilon}{2}\right)^{\lfloor n/2 \rfloor}\left(\frac{\varepsilon}{2}\right)^{\lceil n/2  \rceil }.
 \end{equation}
 Note that this is approximately $(8\varepsilon)^{\lceil n/2  \rceil}$ for large $n$, as one can see by using 
 the Stirling approximation.\qqed
}
\end{proof}

\section{Maximally Biased \dprms{}}
Consider a \prm{} which fulfills the CHSH condition in three out of the four input-cases with
probability $1-\delta$ and in the fourth case perfectly, and where the output bit $X$ is maximally
biased towards zero. 

\begin{definition}
{\rm
 A \emph{maximally biased \dprm{}} is a bipartite conditional probability distribution
 given by the following probability table.  
\begin{eqnarray}
\label{1box_maximally_biased_table}
\begin{array}{c c||c|c||c|c||}
$\backslashbox{V}{U}$& & \multicolumn{2}{c||}{0} & \multicolumn{2}{c||}{1} \\
 & $\backslashbox{Y}{X}$ & 0 & 1 & 0 & 1 \\ \hline\hline
\multirow{2}{*}{0} & 0 & \frac{1}{2}-\frac{\delta}{2} & 0 & \frac{1}{2}-\frac{\delta}{2} & 0 \\ \cline{2-6}
& 1 & \delta & \frac{1}{2}-\frac{\delta}{2} & \delta & \frac{1}{2}-\frac{\delta}{2} \\ \hline\hline
\multirow{2}{*}{1} & 0 & \frac{1}{2}-\frac{\delta}{2} & 0 & 0 & \frac{1}{2}-\frac{\delta}{2} \\ \cline{2-6}
& 1 & \delta & \frac{1}{2}-\frac{\delta}{2} & \frac{1}{2}+\frac{\delta}{2} & 0 \\ \hline\hline
\end{array}
\end{eqnarray}
}
\end{definition}

The local part of one maximally biased \dprm{} is $3\delta$ which can be
reached by the following decomposition.
\begin{small}
\begin{eqnarray}\label{eq:best_dec_1_biased_box}
\nonumber
%\hspace*{-2.5cm}
\fl P^{1,\delta}_{XY|UV}
&=&
+\delta \cdot
\begin{array}{c c||c|c||c|c||}
& & \multicolumn{2}{c||}{0} & \multicolumn{2}{c||}{1} \\
 & 
& 0 & 1 & 0 & 1 \\ \hline\hline
\multirow{2}{*}{0} & 0 & 1 & 0 & 1 & 0 \\ \cline{2-6}
& 1 & 0 & 0 & 0 & 0 \\ \hline\hline
\multirow{2}{*}{1} & 0 & 0 & 0 & 0 & 0 \\ \cline{2-6}
& 1 & 1 & 0 & 1 & 0 \\ \hline\hline
\end{array}
+
\delta \cdot
\begin{array}{c c||c|c||c|c||}
& & \multicolumn{2}{c||}{0} & \multicolumn{2}{c||}{1} \\
 & 
& 0 & 1 & 0 & 1 \\ \hline\hline
\multirow{2}{*}{0} & 0 & 0 & 0 & 0 & 0 \\ \cline{2-6}
& 1 & 1 & 0 & 0 & 1 \\ \hline\hline
\multirow{2}{*}{1} & 0 & 1 & 0 & 0 & 1 \\ \cline{2-6}
& 1 & 0 & 0 & 0 & 0 \\ \hline\hline
\end{array} \\ \nonumber
&&+
\delta \cdot
\begin{array}{c c||c|c||c|c||}
& & \multicolumn{2}{c||}{0} & \multicolumn{2}{c||}{1} \\
 & 
& 0 & 1 & 0 & 1 \\ \hline\hline
\multirow{2}{*}{0} & 0 & 0 & 0 & 0 & 0 \\ \cline{2-6}
& 1 & 0 & 1 & 1 & 0 \\ \hline\hline
\multirow{2}{*}{1} & 0 & 0 & 0 & 0 & 0 \\ \cline{2-6}
& 1 & 0 & 1 & 1 & 0 \\ \hline\hline
\end{array}
%\\ \nonumber
+(1-3\delta)\cdot
\begin{array}{c c||c|c||c|c||}
& & \multicolumn{2}{c||}{0} & \multicolumn{2}{c||}{1} \\
 & 
& 0 & 1 & 0 & 1 \\ \hline\hline
\multirow{2}{*}{0} & 0 & \slantfrac{1}{2} & 0 & \slantfrac{1}{2} & 0 \\ \cline{2-6}
& 1 & 0 & \slantfrac{1}{2} & 0 & \slantfrac{1}{2} \\ \hline\hline
\multirow{2}{*}{1} & 0 & \slantfrac{1}{2} & 0 & 0 & \slantfrac{1}{2} \\ \cline{2-6}
& 1 & 0 & \slantfrac{1}{2} & \slantfrac{1}{2} & 0 \\ \hline\hline
\end{array}
\end{eqnarray}
\end{small}

We will now see that for $n(>1)$ maximally biased \dprms{}, the local part is
$(3\delta)^n$. This value can obviously be reached by decomposing each of
the $n$ individually. 

\begin{lemma}\label{lemma:n_strat_biased_loses_all}
 For every local deterministic strategy for $n$ maximally biased \prms{},
 either the strategy has weight zero or there exist inputs $u,v$ such that
 Alice and Bob lose all the $n$ rounds of the CHSH game.
\end{lemma}
\begin{proof}
{\rm
For the maximally biased \prms{}, there is a pair of necessary  conditions for a strategy to
have non-zero weight: 
\begin{eqnarray}
\label{conda} x_i(u_i=1) &\neq& y_i(v_i=1)\\
\label{condb} (x_i(u_i),y_i(v_i)) &\neq& (1,0)\ \mbox{\ for all\ }  u_i,v_i\neq (1,1)
\end{eqnarray}
On the other hand, if both conditions are fulfilled, then the strategy
{\em can\/} have non-zero weight. Now let us try to construct a strategy with weight
greater than zero --- clearly, only these are of interest for finding the local part. 
We  show that the $i$-th answers to the all-zero and all-one input either
completely determine {\em all} answers for the $i$-th round (and we will be able
to reduce to the case of $n-1$ rounds) --- or they are such that there exists
another input where all rounds are lost.

\medskip

{\sc First case:} {$x_i(1...1)=1$}. Because of~(\ref{conda}) we have 
$y_i(v_i=1)=0$  (no matter what the rest of the input is). Similarly,
we must also have $x_i(u_i=1)=1$, independently of the rest of the input. However, because
of~(\ref{condb}), we must have $y_i(v_i=0)=0$ and $x_i(u_i=0)=0$
independently of the rest of the input. Therefore, all outputs of the $i$-th
round are completely determined by the input of the $i$-th round (giving a
product strategy) and, furthermore, the $i$-th round is lost for the input
$(u_i,v_i)=(1,0)$ for all possible combinations of the remaining inputs, and the
problem reduces to the case of $n-1$ \prms{}.

\medskip

{\sc Second case:} {$x_i(1...1)=0$}. Because of~(\ref{conda}), we have
$y_i(v_i=1)=0$ (no matter what the rest of the input is). And in the same way, we
must also have $x_i(u_i=1)=0$, independently of the rest of the input.
We now classify the \prms{} into two types with respect to Alice's output: Those
for which $x_i(u_i=0)=0\ \mbox{\ holds for all\ }  u$, and those for which there exists $u$ such that 
$u_i=0\land x_i(u)=1$. 
Without loss of generality, we assume that the first $j$ of the $n$ \prms{} are
of the first type, and that the remaining $n-j$ \prms{} are of the second type.
The \prms{} of the first type lose the CHSH game for input $(u_i,v_i)=(0,1)$,
independently of the rest of the input. 
The \prms{} of the second type must yield $y_i(v_i=0)=1$ independently of the
rest of the input --- otherwise the strategy has zero weight. However, this means
that this \prm{} always loses the CHSH game for input $(u_i,v_i)=(1,0)$.
Therefore, all CHSH games are lost for input
$(u,v)$ such that $(u_i,v_i)=(\delta_{i>j},\delta_{i\leq j})$.\qqed
}
\end{proof}

\begin{theorem}\label{theorem:local_part_biased_product}
 The local part of $n$ maximally biased \dprms{} is equal to $(3\delta)^{n}$.
\end{theorem}
\begin{proof}
{\rm
 Lemma~\ref{lemma:n_strat_biased_loses_all} shows that every strategy with non-zero
 weight has at least one input for which all rounds are lost. This shows that the sum of probabilities
of inputs and outputs which lose \emph{all} rounds of the CHSH game
 must be larger or equal the local part. There are $3^n$ of these input/output combinations
 --- each with associated probability $\delta^n$. Thus the
 local part is at most $(3\delta)^n$.
 On the other hand, we can reach a local part of $(3\delta)^n$ by using product strategies.
 It follows that the local part equals $(3\delta)^n$.\qqed
}
\end{proof}

\section{Conclusion}
We have demonstrated that the local part of $n$ isotropic \eprms{} is of order
$\Theta(\varepsilon^{\lceil n/2\rceil})$, and that the local part of $n$ maximally
biased, i.e., maximally non-isotropic, 
 \dprms{} is exactly $(3\delta)^n$.
To exactly quantify the local part of $n$ isotropic \eprms{}  remains an open
problem.

\ack{}

This work has been supported by the Swiss National Science Foundation, by an ETHIIRA grant of ETH's research 
commission,
 and by the National Research Foundation and the Ministry of Education, Singapore. The authors thank two anonymous reviewers
for their valuable comments.

\newpage

\appendix

\section{Further Inspection of the Symmetric Case}\label{sec:further_inspection}

Since this appendix is devoted to symmetric \eprms, for simplicity we write $P^{(\varepsilon)}$ instead of $P^{1,\varepsilon}_{XY|UV}$. As explained in the main text, a known result is
\begin{eqnarray*}
P^{(\varepsilon)}&=&(1-4\varepsilon)\,P^{(0)}\,+\,4\varepsilon\,P^{(1/4)}
\end{eqnarray*}
where $P^{(1/4)}$ is the closest local point to $P^{(0)}$ (actually, $P^{(1/4)}=\demi P^{(0)}+\demi I$ with $I$ the uniform probability distribution $P_{XY|UV}=\frac{1}{4}$). We want to study the non-local part of
\begin{eqnarray*}
\mathbf{P}&=&P^{(\varepsilon)}\times P^{(\varepsilon)}\,=\nonumber\\&=&(1-4\varepsilon)^2\,P^{(0)}\times P^{(0)}\,+\, (4\varepsilon)^2\,P^{(1/4)}\times P^{(1/4)}\nonumber\\&&+\,4\varepsilon(1-4\varepsilon)\,\left[P^{(0)}\times P^{(1/4)}+ P^{(1/4)}\times P^{(0)}\right]\,.
\end{eqnarray*}
We are going to show that $\mathbf{S}\equiv P^{(0)}\times P^{(1/4)}+ P^{(1/4)}\times P^{(0)}=P_{NL}+P_L$. This implies
\begin{eqnarray*}
\mathbf{P}&=&(1-4\varepsilon)\,\mathbf{P_{NL}}\,+\,4\varepsilon\,\mathbf{P_{L}}
\end{eqnarray*}
with $\mathbf{P_{NL}}=(1-4\varepsilon)P^{(0)}\times P^{(0)}+4\varepsilon\,P_{NL}$ and $\mathbf{P_{L}}=(1-4\varepsilon)P_L+4\varepsilon\,P^{(1/4)}\times P^{(1/4)}$; as a consequence, the local part of $\mathbf{P}$ is $4\varepsilon$, just as the local part of the single copy $P^{(\varepsilon)}$.

The most elegant way of finding $P_L$ exploits a symmetry. Indeed, all $P_{XY|UV}$ above the facet $CHSH=2$ can be brought to the form $P^{(\varepsilon)}$ by applying the depolarization procedure $\cal{D}$ defined in Appendix A of \cite{masanes}. For instance, $P^{(1/4)}={\cal{D}}([0\,0; 0\,0])$ where $[0 \,0; 0\, 0]$ is the deterministic probability distribution $P_{XY|UV}=\delta_{X,0}\delta_{Y,0}$ i.e. $X_U=0$ and $Y_V=0$. It is in particular obvious that ${\cal{D}}\times{\cal{D}}[\mathbf{P}]=\mathbf{P}$ and consequently ${\cal{D}}\times{\cal{D}}[\mathbf{S}]=\mathbf{S}$. It is therefore natural to look for $P_L={\cal{D}}\times{\cal{D}}[\mathbf{D}]$ where $\mathbf{D}$ is some deterministic point for four inputs and four outputs. By inspection, one finds
\begin{eqnarray*}
P_L&=&{\cal{D}}\times{\cal{D}}([0 0 0 1; 0 0 2 0])
\end{eqnarray*}
with $\mathbf{D}=[0 0 0 1; 0 0 2 0]$ the deterministic point where the $X_U$ and $Y_V$ are defined as $X_0=X_1=X_2=0, X_3=1, Y_0=Y_1=Y_3=0, Y_2=2$. Finally, since each application of ${\cal{D}}$ defines an orbit of $8$ points, each of $P_L$ and $P^{(1/4)}\times P^{(1/4)}$ is a convex combination of 64 deterministic points; therefore $\mathbf{P_{L}}$ is a convex combination involving 128 deterministic points. The explicit list is given below.

Two remarks to conclude:
\begin{itemize}
\item None of the 64 $4\times 4$ deterministic points, whose mixture gives $P_L$, can be described as a product of two $2\times 2$ deterministic points. For instance, consider Alice in $\mathbf{D}=[0 0 0 1; 0 0 2 0]$, and write both inputs and outputs in binary form: $X_{00}=X_{01}=X_{10}=0\equiv (0,0)$ but $X_{11}=1\equiv(0,1)$. Manifestly, this list cannot be written as $X_{U=uu'}=(x_u, x_{u'})$.

\item One could in principle study the local part of $\mathbf{P_n}=P^{(\varepsilon)}\times...\times P^{(\varepsilon)}$ the product of $n$ symmetric \eprms: indeed,
\begin{eqnarray*}
\mathbf{P_n}&=&\sum_{k=0}^n (4\varepsilon)^k(1-4\varepsilon)^{n-k}\,\mathbf{S_{n,k}}
\end{eqnarray*} with $\mathbf{S_{n,k}}$ the sum of all terms containing $k$ factors $P^{(1/4)}$ and $n-k$ factors $P^{0}$. Obviously, $\mathbf{S_{n,0}}$ is fully non-local and $\mathbf{S_{n,n}}$ is fully local. All the others may contain both a local and a non-local part, satisfying the symmetry ${\cal{D}}\times...\times{\cal{D}}$. Studying the local part of each $\mathbf{S_{n,k}}$ then gives a lower bound on the local part of $\mathbf{P_n}$.  Unfortunately, we have not found an easy way of finding the result. Even for the case $n=3$, the inspection is too heavy; we have evidence that $\mathbf{S_{n,1}}$ should be fully non-local, while $\mathbf{S_{n,2}}$ has a local part.
\end{itemize}

\end{document}